\documentclass{article}

\usepackage{fullpage}
\usepackage{enumerate}
\usepackage{graphicx}
\usepackage{setspace}
\usepackage{listings}
\usepackage[margin=1in]{geometry}
\usepackage{xcolor}
\usepackage{forest}
\usepackage{adjustbox}
\usepackage{float}
\usetikzlibrary{arrows.meta,backgrounds,calc,positioning}

\usepackage{amsfonts}
\usepackage{amssymb}
\usepackage{amsthm}

\newcommand{\CCfont}[1]{\ensuremath{\mathsf{#1}}}
\newcommand{\AC}{\CCfont{AC}}

\newcommand{\ACC}{\CCfont{ACC}}
\newcommand{\NEXP}{\CCfont{NEXP}}
\newcommand{\NQP}{\CCfont{NQP}}

\newtheorem{theorem}{Theorem}[subsection]
\newtheorem{corollary}[theorem]{Corollary}
\newtheorem{lemma}[theorem]{Lemma}

\theoremstyle{definition}

\newtheorem*{definition}{Definition}

\usepackage{hyperref}
\definecolor{linkblue}{HTML}{0055CC}
\hypersetup{colorlinks=true, allcolors=linkblue}

\usepackage[most]{tcolorbox}
\newtcolorbox{openproblem}[1]{
  colback=white, colframe=black!65, colbacktitle=black!6,
  coltitle=black, fonttitle=\bfseries, title={#1},
  boxrule=0.6pt, arc=0pt, left=8pt, right=8pt, top=6pt, bottom=6pt,
  before skip=10pt, after skip=10pt
}

\title{\bf Majority and Torus Polynomials:\\
Bounded-Lift Lower Bounds and Explicit Constructions}
\date{}

\author{
  Saint Wesonga\\
  {\fontsize{10}{12}\selectfont University of Wyoming}\\
  {\fontsize{10}{12}\selectfont\texttt{swesonga@uwyo.edu}}
}

\definecolor{codegray}{rgb}{0.95, 0.95, 0.95}
\begin{document}

\renewcommand*{\theHtheorem}{\thesubsection.\arabic{theorem}}

\maketitle

\begin{abstract}
  Bhrushundi et al. (2019)~\cite{bhlr} introduced torus polynomial approximation
  as an approach to separating Majority from $\ACC^0$. In this model,
  approximation is measured modulo integers. Krishan and Vishwanathan
  (2026)~\cite{kv} developed a duality framework for lower bounds and posed five
  open problems. We resolve two of these, concerning integer lifts and explicit
  polynomial approximations.  An integer lift records the integer shifts that
  turn a torus approximation into an ordinary real approximation. We show that a
  degree-$d$ torus polynomial approximating Majority with error at most $1/8$
  and lift contained in an interval of width $B$ must satisfy
  $d(B+1)=\Omega(n)$. This rules out low-degree approximating polynomials with
  lifts in $\{-1,0,1\}$, answering their Open Problem 2. We also answer Open
  Problem 5 by giving, for every Boolean function on $n\ge1$ variables, an
  explicit torus polynomial approximation of degree at most $n-1$. Its error is
  optimal: zero when the accepting set has even cardinality and $2^{-n-1}$
  otherwise.  The results presented here were discovered and formalized in Lean
  using GPT-6 Astra.
\end{abstract}

\section{Introduction}

Separating Majority from polynomial-size, constant-depth circuits with modular
gates is a central challenge in circuit complexity. Furst, Saxe, and
Sipser~\cite{FuSaSi81} and, independently, Ajtai~\cite{Ajtai1983} showed that
the PARITY function is not in $\AC^0$, the class of Boolean functions computed
by nonuniform polynomial-size, constant-depth circuits over the standard
unbounded-fan-in Boolean basis (AND, OR, NOT). This result was strengthened by
H{\aa}stad~\cite{hastad1986almost}. Subsequent work investigated the power of
including modular gates in the basis. For an integer $q\geq2$, the class
$\AC^0[q]$ consists of functions computed by nonuniform polynomial-size,
constant-depth circuits with MOD$_q$ added to the Boolean
basis. Razborov~\cite{Razb1} proved that Majority is not in $\AC^0[2]$.
Building on this method, Smolensky~\cite{Smo87} proved that for any prime $p$
and any integer $q\geq2$ that is not a power of $p$, the MOD$_q$ function is not
in $\AC^0[p]$. The union of $\AC^0[q]$ over all integers $q\geq2$ is called
$\ACC^0$. There was little progress on obtaining lower bounds against $\ACC^0$
until Williams~\cite{Williams:2014} showed that $\NEXP$ (nondeterministic
exponential time) is not contained in the class.  Murray and
Williams~\cite{Williams:2018} improved this result from $\NEXP$ to $\NQP$
(nondeterministic quasi-polynomial time).  A major open question is whether the
Majority function is in $\ACC^0$.

Bhrushundi, Hosseini, Lovett, and Rao~\cite{bhlr} introduced torus polynomial
approximation as an algebraic approach to $\ACC^0$ lower bounds.  Torus
polynomials accommodate modular gates while retaining ordinary polynomial
degree. In this model, a real polynomial $P$ approximates a Boolean function $f$
on the \emph{Boolean cube} $\{0,1\}^n$. Approximation is measured modulo
integers. At each input $x$, an \emph{integer lift} $Z(x)\in\mathbb Z$ is chosen
so that $|P(x)-Z(x)-f(x)/2|$ is small. The lift therefore turns a modular
approximation into an ordinary real approximation of $Z+f/2$.  Its values can
vary substantially across inputs even when $P$ has low degree.  Bhrushundi et
al.~\cite[Corollary~20]{bhlr} show that $\ACC^0$ circuits can be simulated by
torus polynomials.  If such a circuit computes $f$ on $n$ inputs, then for every
$0<\varepsilon<1/2$ there is a polynomial of degree
$(\log(n/\varepsilon))^{O(1)}$ that approximates $f/2$ modulo integers within
$\varepsilon$ at every Boolean input. In particular, choosing inverse-polynomial
error gives polylogarithmic degree. Consequently, a suitable
super-polylogarithmic torus-degree lower bound for Majority would show that
Majority cannot be computed by $\ACC^0$ circuits.  Bhrushundi et al.~\cite{bhlr}
also prove a lower bound of $\Omega(\sqrt{n/\log n})$ for symmetric torus
polynomials approximating Majority at error $1/(20n)$.

For ordinary real approximation, averaging a polynomial over all permutations of
its inputs preserves both the degree bound and the error when the target is
symmetric, as Majority is. The average is therefore a symmetric approximating
polynomial. For torus approximation, the integer lifts can differ between
permuted inputs, and their average need not be an integer.  For example, $0$ and
$1$ both represent $0$ modulo one, but their average modulo $1$ is $1/2$. Thus,
unlike the real case, one cannot assume without loss of generality that a torus
approximation to Majority is symmetric. Establishing the above lower bound for
nonsymmetric polynomial approximations would show that Majority is not in
$\ACC^0$.

Krishan and Vishwanathan~\cite{kv} recently developed a linear programming
framework for torus lower bounds. Fixing an integer lift turns approximation
into linear constraints on the polynomial coefficients. Their duality theorem
shows that a degree lower bound is equivalent to producing, for every lift, a
witness certifying that these constraints are infeasible. For Majority, they
normalize and bound the lifts, leaving only finitely many candidates at each
input length and degree. This makes it possible to advance towards a lower bound
by excluding successive families of lifts.  Their framework also yields lower
bounds for unrestricted torus polynomials approximating AND and stronger bounds
for symmetric torus polynomials approximating AND and Majority.  In particular,
their AND bounds show that imposing symmetry can increase the required degree,
confirming a substantive obstruction to symmetrization.  Their Open Problem 2
asks whether, for all sufficiently large $n$, one can rule out torus polynomials
of degree $O(\sqrt{n/\log n})$ that approximate Majority with error at most
$1/(20n)$ and whose integer lifts take values only in $\{-1,0,1\}$.  The
following formulation makes the required dual witness precise.

\begin{openproblem}{Open Problem 2 of Krishan and Vishwanathan~\cite{kv}: three-valued integer lifts}
We state a stronger form of the problem, omitting the normalization
$Z(x)=0$ for $|x|\le d(n)$ imposed after Lemma~9 in
Krishan and Vishwanathan~\cite[Section~3]{kv}. Thus the statement includes
all the normalized three-valued lifts considered there.

Let $M_n(x)=\mathbf{1}[|x|\ge n/2]$ be Majority with ties accepted, where $|x|$
counts the ones in $x$. Let $d:\mathbb N\to\mathbb N$ satisfy
$d(n)=O(\sqrt{n/\log n})$. Show that, for all sufficiently large $n$ and every
$Z:\{0,1\}^n\to\{-1,0,1\}$, there is a real function $\gamma:\{0,1\}^n\to\mathbb
R$ such that
\[
 \sum_x\gamma(x)q(x)=0
 \quad\hbox{for every polynomial $q$ with $\deg q\le d(n)$},
\]
and
\[
 \left|\sum_x\gamma(x)\left(Z(x)+\frac{M_n(x)}2\right)\right|
 >\frac1{20n}\sum_x|\gamma(x)|.
\]
All sums range over $\{0,1\}^n$. This is the assertion that every three-valued
lift is \emph{good}, meaning that it has a dual witness excluding a
degree-at-most-$d(n)$ approximation to $Z+M_n/2$ with uniform error
$1/(20n)$. The witness may depend on $Z$.
\end{openproblem}

Krishan and Vishwanathan also prove that every Boolean function has a torus
polynomial approximation of degree at most $n-1$ and error $2^{-n-1}$, symmetric
when the function is symmetric. Their Open Problem 5 asks for an explicit
construction achieving this guarantee.

\begin{openproblem}{Open Problem 5 of Krishan and Vishwanathan~\cite{kv} (paraphrased): an explicit approximation at tiny error}
Given $n\ge1$ and any Boolean function $f:\{0,1\}^n\to\{0,1\}$, explicitly
construct a real polynomial $P$ of degree at most $n-1$ such that
\[
 \min_{z\in\mathbb Z}\left|P(x)-\frac{f(x)}2-z\right|
 \le 2^{-n-1}
 \quad\hbox{for every }x\in\{0,1\}^n.
\]
Krishan and Vishwanathan~\cite{kv} already established the existence
of such a polynomial. The question asks for a construction.
\end{openproblem}

We answer these two questions. For Open Problem 2, we prove the existence of
separating witnesses for every three-valued lift of Majority, extending the
requested degree range to degrees below a constant fraction of $n$.  For Open
Problem 5, we give an explicit torus approximation of optimal error for every
Boolean function.

\subsection{Results and significance}

Our first result converts torus approximation into real approximation for
arbitrary Boolean functions. It uses a \emph{decoder}, a univariate polynomial
$H$ whose output is close to zero near the relevant even integers and close to
one near the relevant odd integers. Since $2P(x)$ is close to $2Z(x)+f(x)$,
whose parity is $f(x)$, the composition $H(2P(x))$ approximately recovers
$f(x)$. For an interval-width-$B$ lift, a decoder of degree $O(B+1)$ converts a
degree-$d$ torus approximation at error at most $1/8$ into a real constant-error
approximation of degree $O((B+1)d)$ (Theorem~\ref{thm:bounded-lift}). Thus a
low-degree torus approximation to any function of real approximate degree
$\Omega(n)$ must satisfy $B+1=\Omega(n/d)$, and $B=\Omega(n/d)$ when
$d=o(n)$. This applies to both Majority and PARITY.

Paturi's theorem~\cite{paturi} gives Majority's linear real approximate-degree
lower bound, with the explicit constant $1/10000$ established by our
formalization. Combining this bound with the conversion rules out the proposed
approximations. Finite-dimensional separation then gives the witnesses
nonconstructively (Corollary~\ref{cor:dual}). For lifts in $\{-1,0,1\}$, taking
$B=2$ gives a real approximation of degree at most $47d$. Since the error
$1/(20n)$ required by Open Problem 2 is at most $1/8$ for every $n\geq1$, this
general conversion already yields the requested separating witnesses whenever
$d<n/470000$.  A degree-five decoder, called a \emph{quintic decoder}, improves
only the numerical constant. It reduces the composition degree from $47d$ to
$5d$, giving the stronger threshold $d<n/50000$ for $n\geq40$.

Our second result gives a closed formula for the polynomial approximations at
tiny error in Open Problem 5 (Theorem~\ref{thm:explicit}). If the top
coefficient of the integer multilinear extension of $f$ is $2k+b$, with
$b\in\{0,1\}$, then subtracting an integer-valued monomial and a multiple of the
alternating-sign polynomial $\chi(x)=\prod_{i=1}^n(2x_i-1)$ leaves degree at
most $n-1$ and error exactly $b2^{-n-1}$. The formula preserves symmetry and
attains the alternating-sum lower bound of Krishan and
Vishwanathan~\cite{kv}. For Majority, this error is positive exactly when $n$ is
a power of two (Corollary~\ref{cor:majority-lengths}).

The two answers clarify different aspects of the approximation model. The first
quantifies the lift width required by functions of large real approximate
degree. PARITY already exhibits this tradeoff, although it belongs to
$\AC^0[2]\subseteq\ACC^0$. Using this framework to separate Majority from
$\ACC^0$ will therefore require properties of Majority beyond its real
approximate degree. The second makes the previously existential tiny-error
approximation theorem constructive and determines its exact error for every
Boolean function.

The results are formalized in Lean~4 using the Mathlib library. The
formalization also includes a proof of Paturi's characterization, with a dual
lower-bound construction following Bun and Thaler~\cite{bt}, and the Majority
approximate-degree lower bound used in the torus arguments.  The Lean sources
are available at
\url{https://github.com/formalcs/circuit-complexity/releases/tag/release%2F0.2.0}.
Appendix~\ref{app:lean} collects the corresponding Lean statements and explains
their notation. References alongside the mathematical arguments point interested
readers to the relevant listings.

\subsection{AI disclosure}

The results presented here were discovered and formalized using GPT-6 Astra.
The model used the Lean circuit-complexity framework developed by
Wesonga~\cite{wesonga26} for the formalization of PARITY lower bounds as a
starting point.  AI assistance was also used to examine related literature,
draft and revise the exposition, and maintain the LaTeX manuscript. The author
guided the choice of questions and the scope of the paper and requested
revisions to its mathematical framing and presentation. The author is
responsible for the mathematical claims, citations, and final text.

\section{Preliminaries}

\subsection{Real and torus approximation}

Our lower bounds compare two ways of approximating a Boolean function by a
polynomial. Both require a single polynomial to work at every input, but real
approximation measures error on the real line, while torus approximation
measures it modulo integers. The integer lift connects these models and lets us
use real approximate-degree lower bounds to study torus polynomials.

Write $[n]=\{1,\ldots,n\}$, $x_S=\prod_{i\in S}x_i$, and $1_S$ for the Boolean
vector with support $S$. All polynomials may be multilinearized using
$x_i^2=x_i$ when evaluated on the Boolean cube. This operation does not increase
degree. We use strict Majority,
\[
 \mathsf{MAJ}_n(x)=\mathbf{1}[|x|>n/2].
\]
For Majority with ties accepted, $M_n(x)=1-\mathsf{MAJ}_n(\bar x)$,
where $\bar x$ is the coordinatewise complement of $x$. Thus the
transformation $P'(y)=1/2-P(\bar y)$ and $Z'(y)=-Z(\bar y)$ transfers
an approximation of strict Majority to one of $M_n$, preserving the degree,
the error, and the lift width, as well as the restriction to lifts in
$\{-1,0,1\}$. Moreover,
$|M_n^{-1}(1)|=2^n-|\mathsf{MAJ}_n^{-1}(1)|$ has the same parity as
$|\mathsf{MAJ}_n^{-1}(1)|$ for $n\geq1$, so the exact-error classification
in Corollary~\ref{cor:majority-lengths} holds for both tie conventions.

Real approximate degree measures the degree needed to achieve a prescribed
uniform error. For a Boolean function $f$, its real $\eta$-approximate degree
$\widetilde{\deg}_{\eta}(f)$ is the least degree of a real polynomial $Q$
satisfying $|Q(x)-f(x)|\le\eta$ at every Boolean point.  At error $1/3$, the
polynomial's values on inputs with $f(x)=0$ and $f(x)=1$ lie in disjoint
intervals. A lower bound on this degree therefore limits how efficiently a
polynomial can distinguish the two outputs.

On the torus $\mathbb T=\mathbb R/\mathbb Z$, real numbers that differ by an
integer represent the same point. In particular, $0$ and $1$ coincide, so we
encode the two Boolean outputs by $0$ and $1/2$ instead. The distance to the
nearest integer,
\[
 \|u\|_{\mathbb T}=\min_{z\in\mathbb Z}|u-z|,
\]
measures the resulting approximation error.

\begin{definition}
A real polynomial $P$ approximates $f$ on the torus with error at most
$\varepsilon$ if
\[
 \max_{x\in\{0,1\}^n}\left\|P(x)-\frac{f(x)}2\right\|_{\mathbb T}
 \le\varepsilon.
\]
Equivalently, there is an integer-valued function $Z$ such that
$|P(x)-Z(x)-f(x)/2|\le\varepsilon$ everywhere. Such a $Z$ is an
\emph{integer lift}. Its interval width is
$\max_x Z(x)-\min_x Z(x)$.
\end{definition}

The lift records which integer translate of $f(x)/2$ the polynomial approaches
at each input. Bounding its interval width limits the spread of these
translates, which is what makes a polynomial decoder possible.  For example, a
lift taking values in $\{-1,0,1\}$ has interval width at most two. Interval
width is different from the number of distinct values of $Z$. An arbitrary
constant integer may be subtracted from both $P$ and $Z$ without changing the
degree or error. Subtracting a polynomial that is integer-valued on the Boolean
cube from both $P$ and $Z$ also preserves the error, but may change both the
degree and the lift width. In that case, our bounds use the degree of the
modified polynomial and the width of the modified lift.

For the explicit-construction question, we instead fix the degree and ask how
small the error can be. We therefore define
\[
 E_d(f)=\inf_{\deg P\le d}
       \max_{x\in\{0,1\}^n}
       \left\|P(x)-\frac{f(x)}2\right\|_{\mathbb T}.
\]
Half the multilinear extension of $f$ gives an exact torus representation of
degree at most $n$, so $E_n(f)=0$. The quantity $E_{n-1}(f)$ measures whether
reducing this degree bound by one forces a positive error, and how large that
error must be. We will determine this quantity exactly and construct a
polynomial attaining it.\footnote{The formal definitions of real and torus
approximation appear in Appendix~\ref{app:lean},
Listing~\ref{lst:approximation-definitions}.}

\subsection{The classical symmetric-function characterization}

Paturi's theorem~\cite{paturi} determines the real approximate degree of any
symmetric Boolean function up to constant factors. The key feature is the
location of the function's changes in value as the Hamming weight increases. A
change near the middle of the Boolean cube forces a larger approximate degree
than changes confined to weights near the endpoints.

A symmetric Boolean function depends only on the number of ones in its
input. Its \emph{profile} is the function $F:\{0,\ldots,n\}\to\{0,1\}$, where
$F(k)$ is the common value of $f$ on inputs of Hamming weight $k$. Thus
$f(x)=F(|x|)$. We write $K_F$ for the set of \emph{transition weights}.
A weight $k$ belongs to this set when the values at
weights $k$ and $k+1$ differ, so
\[
 K_F=\{k\in\{0,\ldots,n-1\}:F(k)\ne F(k+1)\}.
\]
For a nonconstant profile, define
\[
 \Gamma(F)=\min_{k\in K_F}|2k+1-n|,
 \qquad r(F)=\max_{k\in K_F}\min\{k+1,n-k\}.
\]
These parameters describe how close a transition lies to the middle of the
weight grid. The quantity $\Gamma(F)$ is twice the distance from $n/2$ to the
nearest transition midpoint $k+1/2$. For a transition at $k$, there are $k+1$
grid points on its left and $n-k$ on its right.  The parameter $r(F)$ is the
largest size of the smaller of these two sides. Thus a transition near the
middle gives a small $\Gamma(F)$ and a large $r(F)$. If $F(0)=\cdots=F(n)$, then
$f$ is constant and there are no transitions. In this case, set $\Gamma(F)=n$
and $r(F)=0$.

Paturi's characterization states that there are absolute constants $c,C>0$,
independent of $n$ and $F$, such that
\begin{equation}\label{eq:paturi}
 c\sqrt{n(n-\Gamma(F))}\ \leq\ \widetilde{\deg}_{1/3}(f)
 \ \leq\ C\sqrt{n(n-\Gamma(F))}.
\end{equation}
For a nonconstant profile, the identity $n-\Gamma(F)=2r(F)-1$ gives the
equivalent expression $\widetilde{\deg}_{1/3}(f)=\Theta(\sqrt{nr(F)})$.  For
constant functions, both expressions are zero, as required.  Our conventions
also include the zero-variable case.  For Majority on $n\geq1$ variables, the
transition lies at the middle of the weight grid, so $\Gamma(F)\leq1$ and
$r(F)=\lceil n/2\rceil$.  Consequently, Majority has real approximate degree
$\Theta(n)$.  This is the lower bound used in our degree--range tradeoff.

We outline our proof of Paturi's characterization~(\ref{eq:paturi}), beginning
with the reduction to one variable. A symmetric function $f$ has a real
polynomial approximation of degree at most $d$ and error $\varepsilon$ if and
only if there is a univariate real polynomial $q$ of degree at most $d$ such
that
\[
 |q(k)-F(k)|\leq\varepsilon\qquad(0\leq k\leq n).
\]
Here the approximation is required only at integer weights.  For the forward
direction, average the multivariate polynomial over all inputs of each weight
$k$. Averaging a squarefree monomial of degree $s$ gives
$\binom{k}{s}/\binom{n}{s}$, which is a polynomial in $k$ of degree at most
$s$. Averaging also preserves the error bound because $f$ is constant on each
weight layer. For the reverse direction, substitute $|x|=\sum_i x_i$ into $q$,
preserving both the degree and the error.\footnote{The univariate approximation
predicate and the formal symmetrization statement appear in
Appendix~\ref{app:lean}, Listings~\ref{lst:univariate-approximation}
and~\ref{lst:symmetrization}.}

For the upper bound, interpolate the profile smoothly between integer weights
and make the substitution $k=n(1+\cos\theta)/2$. The resulting function
$g(\theta)$ is even and $2\pi$-periodic. Its rate of change is bounded by a
Lipschitz constant $L=O(\sqrt{nr(F)})$, meaning that $|g(\theta)-g(\phi)|\leq
L|\theta-\phi|$.  A Jackson approximation, obtained by convolution with the
normalized fourth power of a Dirichlet kernel, gives an even trigonometric
polynomial of degree at most $4m$ and error $O(L/m)$ for $m\geq1$.  Taking $m$
to be a sufficiently large constant multiple of $\sqrt{nr(F)}$ makes the error
at most $1/3$ for a nonconstant profile.  An even trigonometric polynomial is a
polynomial in $\cos\theta$.  Substituting $\cos\theta=2k/n-1$ therefore gives
the required univariate polynomial of degree $O(\sqrt{nr(F)})$. Constant
profiles are approximated exactly by constant polynomials.

For the lower bound, we follow the explicit dual-polynomial approach of Bun and
Thaler~\cite{bt}. The aim is to assign signed weights to the integer grid so
that their weighted sum against every polynomial below the target degree is
zero, while their weighted sum against $F$ is too large to be consistent with a
good approximation. The construction uses reciprocals of products of differences
between selected grid points.  These weights annihilate low-degree
polynomials. Choosing the grid around a transition that attains $r(F)$
concentrates enough weight near the change in $F$ to obtain a lower bound at a
smaller constant error.  A constant-degree error-reduction step extends this
bound to error $1/3$.  For transitions near the endpoints, a discrete Markov
argument supplies the remaining cases. Together these arguments give
$\widetilde{\deg}_{1/3}(f)=\Omega(\sqrt{nr(F)})$.  The characterization,
including the approximation estimates and boundary cases, is formally verified
in Lean.\footnote{The formal statement appears in Appendix~\ref{app:lean},
Listing~\ref{lst:paturi-characterization}.}

\subsection{The full-cube coefficient}

We next introduce the coefficient that controls our explicit construction of
torus approximations of degree at most $n-1$. Let $n\geq1$ and let
$f:\{0,1\}^n\to\{0,1\}$ be any Boolean function. No symmetry assumption is
needed here. Write $p_f:\mathbb R^n\to\mathbb R$ for its unique multilinear
extension, so that $p_f(x)=f(x)$ at every Boolean input.  Then $p_f/2$ represents
$f/2$ exactly on the torus and has degree at most $n$. A multilinear polynomial
has only one possible monomial of degree $n$, namely $x_{[n]}=x_1\cdots x_n$. To
reduce the degree by one, we therefore need to cancel the coefficient of this
monomial while controlling the error modulo integers.

The coefficients of $p_f$ are obtained by inclusion--exclusion from the values of
$f$ on Boolean inputs,
\begin{equation}\label{eq:mobius}
 p_f(x)=\sum_{S\subseteq[n]}a_Sx_S,
 \qquad
 a_S=\sum_{U\subseteq S}(-1)^{|S|-|U|}f(1_U).
\end{equation}
Every coefficient $a_S$ is an integer because every value $f(1_U)$ is zero or
one. We call the coefficient of the monomial involving all $n$ variables the
\emph{full-cube coefficient} and denote it by
\[
 c_f=a_{[n]}.
\]
Its coefficient in $p_f/2$ is $c_f/2$. Subtracting an integer multiple of
$x_{[n]}$ preserves the values modulo integers on the Boolean cube.  Thus an
even $c_f$ allows this coefficient to be canceled with no torus error. When
$c_f$ is odd, a half-integer coefficient remains after such an integer
correction.  To handle this remaining coefficient, define the alternating-sign
polynomial
\[
 \chi(x)=\prod_{i=1}^n(2x_i-1).
\]
On the Boolean cube, each factor is $-1$ or $1$, so $\chi(x)=(-1)^{n-|x|}$ and
$|\chi(x)|=1$. Thus $\chi$ is $(-1)^n$ times the usual parity character
$(-1)^{|x|}$. This sign convention makes its full-cube coefficient $2^n$. A
small multiple of $\chi$ can therefore cancel the remaining coefficient while
introducing an error of the same magnitude at every input. The same signs also
express $c_f$ as an alternating sum of the values of $f$,
\begin{equation}\label{eq:top-parity}
 c_f=\sum_x\chi(x)f(x),
 \qquad
 c_f\equiv |f^{-1}(1)|\pmod2.
\end{equation}
The congruence holds because $1$ and $-1$ are equal modulo two. Hence whether
$c_f$ is even or odd depends only on the parity of the number of accepting
inputs.  This alternating sum also detects the obstruction to a smaller error.
For every polynomial $Q$ of degree below $n$,
\begin{equation}\label{eq:annihilation}
 \sum_x\chi(x)Q(x)=0.
\end{equation}
Indeed, each monomial omits a variable. Pairing inputs that differ only in that
variable leaves the monomial unchanged and reverses the sign of $\chi$, so their
contributions cancel. Thus the alternating sum vanishes for every candidate of
degree at most $n-1$, while its value on the target $f/2$ is $c_f/2$. Integer
lifts can change this sum only by an integer. When $c_f$ is odd, the resulting
half-integer discrepancy forces a positive approximation error. In
Theorem~\ref{thm:explicit}, we use this identity to prove that the correction
based on $\chi$ attains the smallest possible error.

\section{The degree cost of a bounded integer lift}

This section shows how restricting the range of an integer lift limits the power
of torus approximation. We first convert an approximation with lift width $B$
into an ordinary real approximation, increasing the degree by a factor of
$O(B+1)$. Paturi's theorem then gives a degree lower bound for
Majority. Finally, we sharpen the conversion for lifts taking values in
$\{-1,0,1\}$ and use it to obtain the dual witnesses requested in Open Problem 2
of Krishan and Vishwanathan~\cite{kv}.

\subsection{A conversion to real approximation}

Suppose $P$ approximates $Z+f/2$ and the lift $Z$ takes values between $a$ and
$a+B$. After subtracting $a$ and doubling, the value of $P$ is close to an
integer between $0$ and $2B+1$. This integer is even when $f(x)=0$ and odd when
$f(x)=1$. We therefore seek a univariate polynomial that recovers these two
Boolean values from a small neighborhood of each integer. Composing it with
$2(P-a)$ will give an ordinary real approximation of $f$.

The function $g(t)=(1-\cos(\pi t))/2$ has exactly the desired values at
integers. Lemma~\ref{lem:decoder} replaces $g$ by a polynomial on the interval
containing all relevant inputs. The extra $1/4$ at either end accommodates the
error after doubling. The key estimate is that the decoder's degree grows only
linearly with $B+1$.

\begin{lemma}[A polynomial decoder for bounded integers]\label{lem:decoder}
For every integer $B\ge0$, there is a real polynomial $H_B$ of degree at most
$16(B+1)-1$ such that, on $I_B=[-1/4,\,2B+5/4]$,
\[
 \left|H_B(t)-\frac{1-\cos(\pi t)}2\right|<\frac1{32}.
\]
\end{lemma}

\begin{proof}
Set $m=16(B+1)$ and $t_0=B+1/2$. Let $C_m$ be the Taylor polynomial of
$\cos(\pi t)$ at $t_0$, through degree $m-1$, and put $H_B=(1-C_m)/2$.
For $t\in I_B$ we have $|t-t_0|\le B+3/4\le B+1$. All $m$th derivatives
of $\cos(\pi t)$ have absolute
value at most $\pi^m$. Taylor's theorem and $m!\ge(m/e)^m$ give, since $m\ge16$,
\[
 |\cos(\pi t)-C_m(t)|
 \le\frac{(\pi(B+1))^m}{m!}
 \le\left(\frac{e\pi}{16}\right)^m
 <\left(\frac34\right)^m\le\left(\frac34\right)^{16}<\frac1{16}.
\]
Here $e<3$ and $\pi<4$ suffice. Dividing the error by two proves the
claim.\footnote{The formal statement appears in Appendix~\ref{app:lean},
Listing~\ref{lst:bounded-decoder}.}
\end{proof}

The decoder gives the following conversion for any Boolean function.  Its degree
multiplies the degree of $P$, and its approximation error combines with the
error in $P$ to remain below $1/3$.  The resulting bound depends on the width
$B$ of the lift's range, independently of the interval's location $a$.

\begin{theorem}[Bounded-lift conversion]\label{thm:bounded-lift}
Let $f:\{0,1\}^n\to\{0,1\}$, $a\in\mathbb Z$, $B\in\mathbb N$, and
$0\le\varepsilon\le1/8$. Suppose a real polynomial $P$ of degree at most $d$ and
an integer-valued function $Z$ satisfy, for every $x$,
\[
 a\le Z(x)\le a+B,
 \qquad
 \left|P(x)-Z(x)-\frac{f(x)}2\right|\le\varepsilon.
\]
Then
\begin{equation}\label{eq:degree-range}
 \widetilde{\deg}_{1/3}(f)\le(16(B+1)-1)d.
\end{equation}
\end{theorem}

\begin{proof}
Put $Q=2(P-a)$ and $W=2(Z-a)+f$. The integer $W(x)$ lies in $[0,2B+1]$, has
parity $f(x)$, and satisfies $|Q(x)-W(x)|\le2\varepsilon\le1/4$.  Thus $Q(x)\in
I_B$. The function $g(t)=(1-\cos(\pi t))/2$ agrees with parity at integer
points. For any integer $w$ and real $h$,
\[
 |g(w+h)-g(w)|
 =\frac{1-\cos(\pi h)}2
 \le\frac{\pi^2h^2}{4}.
\]
Applying Lemma~\ref{lem:decoder} gives
\[
 |H_B(Q(x))-f(x)|
 <\frac1{32}+\pi^2\varepsilon^2
 \le\frac1{32}+\frac{\pi^2}{64}
 <\frac9{32}<\frac13.
\]
The composition has degree at most $(16(B+1)-1)d$.  Multilinearization preserves
its Boolean values, proving (\ref{eq:degree-range}).\footnote{The formal
statement appears in Appendix~\ref{app:lean},
Listing~\ref{lst:bounded-conversion}.}
\end{proof}

For Majority, the real approximation produced by Theorem~\ref{thm:bounded-lift}
must have degree $\Omega(n)$ by Paturi's theorem. Thus the original degree and
the lift width cannot both be small. The next corollary makes this tradeoff
explicit, even when the allowed torus error is the constant $1/8$.

\begin{corollary}[Majority degree--range tradeoff]\label{cor:range}
There is an absolute constant $c>0$ such that every approximation in
Theorem~\ref{thm:bounded-lift} with $f=\mathsf{MAJ}_n$ satisfies
\[
 d(B+1)\ge cn.
\]
In particular, constant-width lifts require degree $\Omega(n)$ even at constant
error $1/8$.
\end{corollary}

\begin{proof}
Paturi's theorem gives
$\widetilde{\deg}_{1/3}(\mathsf{MAJ}_n)=\Omega(n)$~\cite{paturi}.
Apply Theorem~\ref{thm:bounded-lift}.
\end{proof}

Our Lean formalization proves the Majority approximate-degree lower bound with
explicit constant $1/10000$, so the degree--range tradeoff holds with
$c=1/160000$.  For fixed $B$, this gives $d\ge n/(160000(B+1))$. These constants
are not optimized.\footnote{The formal statements appear in
Appendix~\ref{app:lean}, Listings~\ref{lst:paturi-input}
and~\ref{lst:majority-range}.}

The tradeoff excludes any such approximation with $d(B+1)<cn$, whether $B$ is
fixed or grows with $n$. For fixed $B$, the linear order of the degree lower
bound is best possible.  The ordinary multilinear extension of $f/2$ is an exact
degree-at-most-$n$ torus representation with lift zero.

\noindent\textbf{Remark (Sharpness).}
For PARITY, $P(x)=|x|/2$ is an exact torus representation with lift
$Z(x)=\lfloor|x|/2\rfloor$. Here $d(B+1)=\lfloor n/2\rfloor+1$, so the
general tradeoff is sharp up to constants.

\subsection{Three-valued lifts and dual witnesses}

Open Problem 2 concerns lifts taking values only in $\{-1,0,1\}$ and error at
most $1/(20n)$. Such a lift has width two, so the general conversion already
gives a real approximation of degree at most $47d$. Because the lift takes only
three values, $2Z(x)+f(x)$ belongs to the six-point set $\{-2,-1,0,1,2,3\}$.
Interpolating parity exactly on these points gives a decoder of degree five.
Bounding its derivative controls the error near these points, yielding a real
approximation of degree at most $5d$.  Theorem~\ref{thm:ternary} records this
improvement for Majority.

\begin{theorem}[A quintic decoder for the ternary lift]\label{thm:ternary}
Suppose $n\ge40$, and a real polynomial $P$ with $\deg P\le d$ and a function
$Z:\{0,1\}^n\to\{-1,0,1\}$ satisfy
\[
 \left|P(x)-Z(x)-\frac{\mathsf{MAJ}_n(x)}2\right|
 \le\frac1{20n}
 \quad\hbox{for all }x.
\]
Then $\widetilde{\deg}_{1/3}(\mathsf{MAJ}_n)\le5d$.
\end{theorem}

\begin{proof}
Let $Q=2P$ and $W=2Z+\mathsf{MAJ}_n$. Define
\[
 H(t)=\frac{2t^5-5t^4-10t^3+20t^2+8t}{15}.
\]
Its values at $-2,-1,0,1,2,3$ are respectively $0,1,0,1,0,1$.
Thus $H(W(x))=\mathsf{MAJ}_n(x)$.
Since $|Q-W|\le1/(10n)$, all relevant values lie in $[-2.1,3.1]$.
On that interval,
\[
 |H'(t)|\le
 \frac{10(3.1)^4+20(3.1)^3+30(3.1)^2+40(3.1)+8}{15}
 =\frac{646547}{5000}<130.
\]
The mean-value theorem gives
$|H(Q(x))-\mathsf{MAJ}_n(x)|\le13/n<1/3$.
The multilinearization of $H\circ Q$ has degree at most $5d$.
\end{proof}

The same argument applies to any Boolean function $f$, giving
$\widetilde{\deg}_{1/3}(f)\le5d$ under the same approximation
hypotheses.\footnote{The formal statement appears in Appendix~\ref{app:lean},
Listing~\ref{lst:ternary-decoder}.}

Theorem~\ref{thm:ternary} rules out the proposed torus approximation whenever
$5d<\widetilde{\deg}_{1/3}(\mathsf{MAJ}_n)$.  To answer the open problem in its
stated form, we turn this nonexistence result into a dual witness for each fixed
lift $Z$.  Let $V_d$ be the real subspace of functions on $\{0,1\}^n$
represented by degree-at-most-$d$ polynomials. For real functions $u,v$ on the
cube, write $\langle u,v\rangle=\sum_xu(x)v(x)$ and $\|u\|_1=\sum_x|u(x)|$.

A witness $\gamma\in V_d^\perp$ assigns signed weights to Boolean inputs so that
$\langle Q,\gamma\rangle=0$ for every $Q\in V_d$.  If such a $Q$ approximated
the target $Z+\mathsf{MAJ}_n/2$ with pointwise error at most $1/(20n)$, then the
absolute weighted sum of that target would be at most $\|\gamma\|_1/(20n)$.  The
strict inequality in Corollary~\ref{cor:dual} certifies that this is
impossible. Finite-dimensional separation guarantees the existence of a witness
once the target is too far from $V_d$.

\begin{corollary}[Finite witnesses for three-valued lifts]\label{cor:dual}
If $n\ge40$ and $5d<\widetilde{\deg}_{1/3}(\mathsf{MAJ}_n)$, then for
every $Z:\{0,1\}^n\to\{-1,0,1\}$ there is a real function
$\gamma\in V_d^\perp$ such that
\[
 \left|\left\langle Z+\frac{\mathsf{MAJ}_n}{2},\gamma\right\rangle\right|
 >\frac1{20n}\|\gamma\|_1.
\]
\end{corollary}

\begin{proof}
The distance from $Z+\mathsf{MAJ}_n/2$ to $V_d$ in the maximum norm is attained
because the ambient space is finite-dimensional and $V_d$ is closed.  If this
minimum were at most $1/(20n)$, a minimizing polynomial would, by
Theorem~\ref{thm:ternary}, imply $\widetilde{\deg}_{1/3}(\mathsf{MAJ}_n)\le5d$,
contrary to the hypothesis.  Thus the distance is strictly greater than
$1/(20n)$. Separating the target from the closed convex sum of $V_d$ and the
maximum-norm ball of radius $1/(20n)$ yields a nonzero functional $\gamma$
annihilating $V_d$.  The support function of that ball is $\|\gamma\|_1/(20n)$,
which gives the displayed inequality.
\end{proof}

The finite witness conclusion also holds for any Boolean function $f$ satisfying
$5d<\widetilde{\deg}_{1/3}(f)$, with $n\ge40$ as above.\footnote{The formal
statement appears in Appendix~\ref{app:lean}, Listing~\ref{lst:ternary-dual}.}

To obtain the degree range in Open Problem 2, we combine the finite corollary
with the explicit linear lower bound for Majority established in our
formalization of Paturi's theorem.

\begin{corollary}[Asymptotic witnesses for three-valued lifts]\label{cor:asymptotic-dual}
Let $d:\mathbb N\to\mathbb N$ satisfy $50000d(n)<n$ for all sufficiently
large $n$. Then, for all sufficiently large $n$ and every
$Z:\{0,1\}^n\to\{-1,0,1\}$, there exists $\gamma\in V_{d(n)}^\perp$ with
\[
 \left|\left\langle Z+\frac{\mathsf{MAJ}_n}{2},\gamma\right\rangle\right|
 >\frac1{20n}\|\gamma\|_1.
\]
In particular, this holds whenever $d(n)=O(\sqrt{n/\log n})$.
Both conclusions also hold for Majority with ties accepted.
\end{corollary}

\begin{proof}
For every $n\ge40$ satisfying $50000d(n)<n$, the explicit lower bound gives
\[
 5d(n)<\frac{n}{10000}
 \le\widetilde{\deg}_{1/3}(\mathsf{MAJ}_n).
\]
Corollary~\ref{cor:dual} therefore supplies the required witness.  If
$d(n)=O(\sqrt{n/\log n})$, then $d(n)/n\to0$, since $\sqrt{n/\log
  n}/n=1/\sqrt{n\log n}\to0$. Hence $50000d(n)<n$ for all sufficiently large
$n$.  For ties-accepted Majority, the complement transformation with
$Z(x)=-Z'(\bar x)$ rules out approximating $Z'+M_n/2$ within $1/(20n)$ by an
element of $V_{d(n)}$, for every three-valued $Z'$.  The attainment and
separation argument of Corollary~\ref{cor:dual} then gives the required witness.
\end{proof}

The asymptotic statements are formalized for both tie conventions.\footnote{See
Appendix~\ref{app:lean}, Listings~\ref{lst:asymptotic-dual}
and~\ref{lst:linear-dual}.}  More generally, Theorem~\ref{thm:bounded-lift} and
the same separation argument apply at error $1/8$. For any Boolean function $f$
and integer lift $Z$ with values in $[a,a+B]$, where $a\in\mathbb Z$ and
$B\in\mathbb N$, if $(16(B+1)-1)d<\widetilde{\deg}_{1/3}(f)$, there exists
$\gamma\in V_d^\perp$ such that
\[
 \left|\left\langle Z+\frac f2,\gamma\right\rangle\right|
 >\frac18\|\gamma\|_1.
\]
This witness excludes uniform approximation to $Z+f/2$ by degree-at-most-$d$
polynomials at every error $\varepsilon\le1/8$.

\section{Explicit optimal approximation below full degree}

We now turn to Open Problem 5 of Krishan and Vishwanathan~\cite{kv}, which asks
for an explicit construction of torus approximations with exponentially small
error. Their Theorem 18 establishes that every Boolean function on $n\ge1$
variables admits an approximation of degree at most $n-1$ and error at most
$2^{-n-1}$. We give a formula that attains the smallest possible error for each
function. We then describe how to compute its coefficients from the truth table
and determine exactly which input lengths allow an exact representation of
Majority below full degree.

\subsection{Construction and optimality}

The construction implements the coefficient cancellation described in the
full-cube coefficient subsection. Starting with the exact representation $p_f/2$,
write its coefficient of $x_{[n]}$ as $k+b/2$, where $k$ is an integer and
$b\in\{0,1\}$. Subtracting $kx_{[n]}$ removes the integer part without changing
any value on the torus.  If $b=1$, subtracting $\chi/2^{n+1}$ removes the
remaining half while changing each Boolean value by only $2^{-n-1}$ in absolute
value.  If $b=0$, this second correction is unnecessary.

Theorem~\ref{thm:explicit} states the resulting formula and its optimality. The
alternating-sum obstruction from Lemma 14 of Krishan and Vishwanathan~\cite{kv},
applied to the full cube, gives the matching lower bound. Thus the new
construction attains the known obstruction. Because both correction terms are
invariant under coordinate permutations, it also preserves symmetry whenever $f$
is symmetric.

\begin{theorem}[An explicit polynomial approximation of degree at most $n-1$]\label{thm:explicit}
Let $n\ge1$, and let $p_f$ be the multilinear extension (\ref{eq:mobius}) of a
Boolean function $f$. Write $c_f=2k+b$, where $k\in\mathbb Z$ and $b\in\{0,1\}$.
Then
\begin{equation}\label{eq:explicit-new}
 P_f(x)=\frac{p_f(x)}2-kx_{[n]}-\frac{b}{2^{n+1}}
                       \prod_{i=1}^n(2x_i-1)
\end{equation}
has degree at most $n-1$ and satisfies, at every Boolean point,
\[
 \left\|P_f(x)-\frac{f(x)}2\right\|_{\mathbb T}
 =\frac{b}{2^{n+1}}.
\]
Moreover,
\begin{equation}\label{eq:minimax}
 E_{n-1}(f)=
 \left\{\begin{array}{ll}
 0,& |f^{-1}(1)|\hbox{ is even},\\[2pt]
 2^{-n-1},& |f^{-1}(1)|\hbox{ is odd}.
 \end{array}\right.
\end{equation}
If $f$ is symmetric, then $P_f$ is symmetric.
\end{theorem}

\begin{proof}
The coefficient of $x_{[n]}$ in (\ref{eq:explicit-new}) is
$c_f/2-k-b/2=0$. Every other monomial has degree below $n$.
At Boolean points,
\[
 P_f(x)-\frac{f(x)}2=-kx_{[n]}-\frac{b}{2^{n+1}}\chi(x).
\]
The first term is integer-valued and $\chi(x)\in\{-1,1\}$.  This proves the
pointwise error formula.

For completeness, suppose $b=1$ and a polynomial $Q$ of degree at most $n-1$ has
torus error at most $\varepsilon$. Choose an integer lift $Z$ for $Q$. By
(\ref{eq:annihilation}),
\[
 \left|\sum_x\chi(x)\left(Q(x)-Z(x)-\frac{f(x)}2\right)\right|
 =\left|\sum_x\chi(x)Z(x)+\frac{c_f}{2}\right|\ge\frac12.
\]
The left side is at most $2^n\varepsilon$. Thus $\varepsilon\ge2^{-n-1}$. When
$b=0$, the explicit construction has zero error. Equation (\ref{eq:top-parity})
gives the accepting-set formulation in (\ref{eq:minimax}).  Finally, the
uniqueness of $p_f$ implies its symmetry whenever $f$ is symmetric, and both
correction terms are invariant under coordinate permutations.\footnote{The
formal statements appear in Appendix~\ref{app:lean},
Listings~\ref{lst:explicit-formula}, \ref{lst:explicit-degree-error},
and~\ref{lst:explicit-optimal-symmetric}.}
\end{proof}

The theorem shows that reducing the degree from $n$ to $n-1$ has only two
possible effects on the optimal error. An even number of accepting inputs
permits an exact torus representation. An odd number forces error $2^{-n-1}$,
and the formula attains that error at every input. In particular, symmetry
introduces no additional cost for this construction.

To make the construction usable from a truth table, we can expand the formula
and compute every coefficient directly. The full-cube coefficient is zero by
construction. For $S\subsetneq[n]$, the coefficient of $x_S$ is
\[
 \frac{a_S}{2}
 -\frac{b(-1)^{n-|S|}}{2^{n-|S|+1}}.
\]
The Boolean M\"obius transform computes all the $a_S$ by applying successive
differences along each of the $n$ input coordinates.  Each pass uses $2^{n-1}$
integer subtractions, giving exactly $n2^{n-1}$ subtractions in total. The
displayed correction then yields all coefficients of $P_f$ in $O(n2^n)$
arithmetic operations.  This measures the work relative to a truth table with
$2^n$ entries.  The expanded polynomial can itself have exponentially many
coefficients.

The coefficients also have bounded exact representations. A common denominator
is $2^{n+1}$, and each integer numerator has absolute value at most
$2^{2n+1}$. Thus each coefficient can be stored using $O(n)$ bits for its
numerator and denominator. The operation count above treats each integer
subtraction as one operation, regardless of the number of bits in its
operands.\footnote{The formal statements appear in Appendix~\ref{app:lean},
Listings~\ref{lst:coefficient-size} and~\ref{lst:fast-mobius}.}

Two small examples illustrate the two error cases. Two-bit AND has one accepting
input, so its optimal degree-one error is $1/8$.  The construction gives
\[
 P(x_1,x_2)=\frac{x_1+x_2}{4}-\frac18
\]
which attains this error at all four inputs. Two-bit parity has two accepting
inputs, so the construction gives an exact degree-one torus representation,
$(x_1+x_2)/2$. At the input $(1,1)$, its value is $1$, which equals the required
value $0$ modulo integers.

\subsection{The exact length dependence for Majority}

For Majority, Theorem~\ref{thm:explicit} reduces the question of exact
representation below degree $n$ to the parity of its accepting-set size.  The
full-cube coefficient gives a convenient way to determine this parity. For a
threshold function, it is a signed binomial coefficient, whose parity can be
read from the binary digits of its indices.  For strict Majority, this
calculation singles out precisely the power-of-two input lengths.

The following corollary therefore describes exactly when degree $n-1$ suffices
for zero error. At power-of-two lengths, the formula still attains error
$2^{-n-1}$, but any improvement in error requires restoring full degree $n$.

\begin{corollary}\label{cor:majority-lengths}
For $n\ge1$,
\[
 E_{n-1}(\mathsf{MAJ}_n)=
 \left\{\begin{array}{ll}
 2^{-n-1},&n\hbox{ is a power of two},\\[2pt]
 0,&\hbox{otherwise}.
 \end{array}\right.
\]
At power-of-two lengths, every torus polynomial approximating Majority with
error strictly below $2^{-n-1}$ has degree at least $n$.
\end{corollary}

\begin{proof}
For the threshold $f_t(x)=\mathbf{1}[|x|\ge t]$, Pascal's identity and
cancellation in the alternating sum give
\[
 c_{f_t}=\sum_{j=t}^n(-1)^{n-j}\binom nj
        =(-1)^{n-t}\binom{n-1}{t-1}.
\]
For strict Majority take $t=\lfloor n/2\rfloor+1$.  If $n=2m+1>1$, the binomial
coefficient is $\binom{2m}{m}$, which is even because it equals
$2\binom{2m-1}{m-1}$.  For $n=1$ it is one.

If $n=2m$, the coefficient has the parity of $\binom{2m-1}{m}$.  The
binary-digit criterion for binomial coefficients says that $\binom{a+b}{a}$ is
odd precisely when $a$ and $b$ have disjoint binary supports. This follows by
expanding $(1+z)^u=\prod_{j:\,u_j=1}(1+z^{2^j})$ over $\mathbb F_2$.  With $a=m$
and $b=m-1$, disjointness holds precisely when $m$ is a power of two. Apply
Theorem~\ref{thm:explicit}.\footnote{The formal statement appears in
Appendix~\ref{app:lean}, Listing~\ref{lst:majority-lengths}.}
\end{proof}

At power-of-two input lengths, $2^{-n-1}$ is therefore the smallest error
achievable by a torus polynomial of degree at most $n-1$ approximating
Majority. This error is far below the inverse-polynomial error used in the
circuit simulation.

For Majority, the polynomial $P_f$ in (\ref{eq:explicit-new}) has integer lift
$Z(x)=-kx_{[n]}$. This lift takes values in $\{0,-k\}$, so it has at most two
values but interval width $|k|$. Since $k=(c_f-b)/2$, the binomial coefficient
formula above and Stirling's formula give
\[
 |k|=\Theta\!\left(\frac{2^n}{\sqrt n}\right).
\]
Thus even a lift taking only two values can have exponentially large width.
This illustrates the distinction between interval width and the number of
lift values.

\clearpage
\bibliographystyle{alphaurl}
\bibliography{MajorityAndTorus}

\begin{thebibliography}{BHLR19}

\bibitem[Ajt83]{Ajtai1983}
Mikl{\'o}s Ajtai.
\newblock {$\Sigma^1_1$}-formulae on finite structures.
\newblock {\em Annals of Pure and Applied Logic}, 24(1):1--48, 1983.
\newblock \href {https://doi.org/10.1016/0168-0072(83)90038-6}
  {\path{doi:10.1016/0168-0072(83)90038-6}}.

\bibitem[BHLR19]{bhlr}
Abhishek Bhrushundi, Kaave Hosseini, Shachar Lovett, and Sankeerth Rao.
\newblock Torus polynomials: An algebraic approach to {ACC} lower bounds.
\newblock In {\em 10th Innovations in Theoretical Computer Science Conference
  (ITCS 2019)}, volume 124 of {\em Leibniz International Proceedings in
  Informatics (LIPIcs)}, pages 13:1--13:16, 2019.
\newblock \href {https://doi.org/10.4230/LIPIcs.ITCS.2019.13}
  {\path{doi:10.4230/LIPIcs.ITCS.2019.13}}.

\bibitem[BT15]{bt}
Mark Bun and Justin Thaler.
\newblock Dual lower bounds for approximate degree and {Markov--Bernstein}
  inequalities.
\newblock {\em Information and Computation}, 243:2--25, 2015.
\newblock \href {https://doi.org/10.1016/j.ic.2014.12.003}
  {\path{doi:10.1016/j.ic.2014.12.003}}.

\bibitem[FSS81]{FuSaSi81}
M.~Furst, J.~Saxe, and M.~Sipser.
\newblock Parity, circuits, and the polynomial time hierarchy.
\newblock In {\em Proceedings of the Twenty-second Annual IEEE Symposium on
  Foundations of Computer Science}, pages 260--270, 1981.

\bibitem[H{\aa}s86]{hastad1986almost}
Johan H{\aa}stad.
\newblock Almost optimal lower bounds for small depth circuits.
\newblock In {\em Proceedings of the eighteenth annual ACM symposium on Theory
  of computing}, pages 6--20, 1986.

\bibitem[KV26]{kv}
Vaibhav Krishan and Sundar Vishwanathan.
\newblock Lower bounds and separations for torus polynomials.
\newblock In {\em 17th Innovations in Theoretical Computer Science Conference
  (ITCS 2026)}, volume 362 of {\em Leibniz International Proceedings in
  Informatics (LIPIcs)}, pages 88:1--88:20, 2026.
\newblock \href {https://doi.org/10.4230/LIPIcs.ITCS.2026.88}
  {\path{doi:10.4230/LIPIcs.ITCS.2026.88}}.

\bibitem[MW18]{Williams:2018}
Cody Murray and Ryan Williams.
\newblock Circuit lower bounds for nondeterministic quasi-polytime: an easy
  witness lemma for {NP} and {NQP}.
\newblock In {\em Proceedings of the 50th Annual ACM SIGACT Symposium on Theory
  of Computing}, STOC 2018, pages 890--901, New York, NY, USA, 2018.
  Association for Computing Machinery.
\newblock \href {https://doi.org/10.1145/3188745.3188910}
  {\path{doi:10.1145/3188745.3188910}}.

\bibitem[Pat92]{paturi}
Ramamohan Paturi.
\newblock On the degree of polynomials that approximate symmetric {Boolean}
  functions (preliminary version).
\newblock In {\em Proceedings of the Twenty-Fourth Annual ACM Symposium on
  Theory of Computing (STOC 1992)}, pages 468--474, 1992.
\newblock \href {https://doi.org/10.1145/129712.129758}
  {\path{doi:10.1145/129712.129758}}.

\bibitem[Raz87]{Razb1}
A.~A. Razborov.
\newblock Lower bounds on the size of bounded-depth networks over a complete
  basis with logical addition (in {R}ussian).
\newblock {\em MatZametki}, 41(4):598--607, 1987.
\newblock An English translation appeared in Mathematical Notes of the Academy
  of Sci. of the USSR 41(4) (1987) 333--338.
\newblock URL: \url{http://people.cs.uchicago.edu/~razborov/files/depth.pdf}.

\bibitem[Smo87]{Smo87}
Roman Smolensky.
\newblock Algebraic methods in the theory of lower bounds for {B}oolean circuit
  complexity.
\newblock In {\em Proc. 19th STOC}, pages 77--82, New York, 1987. ACM Press.

\bibitem[Wes26]{wesonga26}
Saint Wesonga.
\newblock Formalizing {PARITY} circuit lower bounds in {Lean}.
\newblock arXiv preprint arXiv:2609.24188, 2026.
\newblock URL: \url{https://arxiv.org/abs/2609.24188}.

\bibitem[Wil14]{Williams:2014}
Ryan Williams.
\newblock Nonuniform {ACC} circuit lower bounds.
\newblock {\em J. ACM}, 61(1):2:1--2:32, January 2014.
\newblock URL: \url{http://doi.acm.org/10.1145/2559903}, \href
  {https://doi.org/10.1145/2559903} {\path{doi:10.1145/2559903}}.

\end{thebibliography}

\clearpage
\appendix
\section{Formal theorem statements in Lean}\label{app:lean}
\renewcommand{\thelstlisting}{\thesection.\arabic{lstlisting}}
\setcounter{lstlisting}{0}

This appendix collects the formal definitions and theorem statements
corresponding to the results in the main text. The listings follow the order of
the mathematical development and retain references to the associated lemmas,
theorems, and corollaries.

\subsection{Definitions and reading conventions}

The numbered code blocks below give Lean definitions and theorem statements
corresponding to the mathematics in the text. Proof scripts are
omitted. Listing~\ref{lst:approximation-definitions} gives the formal
definitions of real and torus approximation.

In these blocks, \lstinline|Fin n| indexes the $n$ input coordinates,
\lstinline|Bool| is the type of Boolean values, and \lstinline|Cube n|
represents the Boolean cube. The function \lstinline|bit| converts false and
true to the real numbers $0$ and $1$, and \lstinline|P x| means $P(x)$.
Membership in \lstinline|degreeSpace (K := Real) (Fin n) d| means that a cube
function is represented by a multilinear polynomial of degree at most $d$.

\noindent\begin{minipage}{\linewidth}
\begin{lstlisting}[caption={Boolean inputs and approximation predicates},
                   mathescape=true, label={lst:approximation-definitions}]
abbrev Cube (n : Nat) := Fin n $\to$ Bool

def bit (b : Bool) : $\mathbb{R}$ := if b then 1 else 0

def HasRealApproximation {n : Nat} (f : Cube n $\to$ Bool) (d : Nat) ($\varepsilon$ : $\mathbb{R}$) : Prop :=
  $\exists$ P : Cube n $\to$ $\mathbb{R}$,
    P $\in$ degreeSpace (K := $\mathbb{R}$) (Fin n) d $\land$
    $\forall$ x, |P x - bit (f x)| $\leq$ $\varepsilon$

def HasTorusApproximation {n : Nat} (f : Cube n $\to$ Bool) (d : Nat) ($\varepsilon$ : $\mathbb{R}$) : Prop :=
  $\exists$ P : Cube n $\to$ $\mathbb{R}$,
    P $\in$ degreeSpace (K := $\mathbb{R}$) (Fin n) d $\land$
    $\exists$ Z : Cube n $\to$ $\mathbb{Z}$, $\forall$ x, |P x - (Z x : $\mathbb{R}$) - bit (f x) / 2| $\leq$ $\varepsilon$
\end{lstlisting}
\end{minipage}

The names \lstinline|realApproxDegree f eta|, \lstinline|optimalTorusError f d|,
and \lstinline|torusDistance u| denote $\widetilde{\deg}_{\eta}(f)$, $E_d(f)$,
and $\|u\|_{\mathbb T}$, respectively. The predicate
\lstinline|HasTorusApproximation f d eps| packages the existence of a
degree-at-most-$d$ approximating polynomial and an integer lift. The definition
\lstinline|strictMajority n| uses the strict tie convention $|x|>n/2$. All
results displayed below have proofs formally verified in Lean. The Majority
degree--range and asymptotic witness consequences use a proved linear real
approximate-degree lower bound in Listing~\ref{lst:paturi-input}.

\subsection{Symmetrization and Paturi's characterization}

The predicate \lstinline|HasUnivariateApproximation n F d eps| in
Listing~\ref{lst:univariate-approximation} means that some real polynomial $q$
of degree at most $d$ satisfies $|q(k)-F(k)|\leq\varepsilon$ for every integer
$0\leq k\leq n$. Here \lstinline|q.eval| denotes evaluation of $q$, and
\lstinline|q.natDegree| is its degree, with the zero polynomial assigned degree
zero.

\noindent\begin{minipage}{\linewidth}
\begin{lstlisting}[caption={Univariate approximation on the Hamming-weight grid},
                   mathescape=true, label={lst:univariate-approximation}]
def HasUnivariateApproximation (n : Nat) (F : Nat $\to$ Bool)
    (d : Nat) ($\varepsilon$ : $\mathbb{R}$) : Prop :=
  $\exists$ q : Polynomial $\mathbb{R}$, q.natDegree $\leq$ d $\land$
    $\forall$ k : Nat, k $\leq$ n $\to$ |q.eval (k : $\mathbb{R}$) - bit (F k)| $\leq$ $\varepsilon$
\end{lstlisting}
\end{minipage}

In Listing~\ref{lst:symmetrization}, \lstinline|profileFunction F| denotes
$x\mapsto F(|x|)$. The theorem equates real approximation of this symmetric
function with univariate approximation of its profile.

\noindent\begin{minipage}{\linewidth}
\begin{lstlisting}[caption={Exact symmetrization of polynomial approximation},
                   mathescape=true, label={lst:symmetrization}]
theorem hasRealApproximation_profile_iff_univariate
    {n d : Nat} {F : Nat $\to$ Bool} {$\varepsilon$ : $\mathbb{R}$} :
    HasRealApproximation (profileFunction (n := n) F) d $\varepsilon$ $\leftrightarrow$
      HasUnivariateApproximation n F d $\varepsilon$
\end{lstlisting}
\end{minipage}

The proof of Paturi's characterization described in the main text is formalized
in Lean, including the kernel and grid estimates, error reduction, and boundary
cases. Listing~\ref{lst:paturi-characterization} gives the resulting
statement. Here \lstinline|paturiScale n F| is $\sqrt{n(n-\Gamma(F))}$ with the
convention for constant profiles used in the main text. Every symmetric Boolean
function has a canonical Hamming-weight profile, so this formulation covers all
such functions.

\noindent\begin{minipage}{\linewidth}
\begin{lstlisting}[caption={Paturi's characterization with uniform absolute constants},
                   mathescape=true, label={lst:paturi-characterization}]
theorem paturi_characterization_profile :
    $\exists$ c C : $\mathbb{R}$, 0 < c $\land$ 0 < C $\land$ $\forall$ (n : Nat) (F : Nat $\to$ Bool),
      c * paturiScale n F $\leq$
        (realApproxDegree (profileFunction (n := n) F) (1 / 3) : $\mathbb{R}$) $\land$
      (realApproxDegree (profileFunction (n := n) F) (1 / 3) : $\mathbb{R}$) $\leq$
        C * paturiScale n F
\end{lstlisting}
\end{minipage}

\subsection{Bounded lifts and dual witnesses}

Listing~\ref{lst:bounded-decoder} gives the equivalent existence statement.
Here \lstinline|Icc| denotes a closed interval.

\noindent\begin{minipage}{\linewidth}
\begin{lstlisting}[caption={Lemma~\ref{lem:decoder}: bounded integer decoder},
                   mathescape=true, label={lst:bounded-decoder}]
theorem exists_bounded_integer_decoder (B : Nat) :
    $\exists$ H : Polynomial $\mathbb{R}$, H.natDegree $\leq$ 16 * (B + 1) - 1 $\land$
      $\forall$ t $\in$ Icc (-1 / 4 : $\mathbb{R}$) (2 * (B : $\mathbb{R}$) + 5 / 4),
        |H.eval t - (1 - Real.cos (Real.pi * t)) / 2| < 1 / 32
\end{lstlisting}
\end{minipage}

The Lean formulation in Listing~\ref{lst:bounded-conversion} states the
same degree bound. Nonnegativity of the error follows from
\lstinline|herr| at any Boolean input.

\noindent\begin{minipage}{\linewidth}
\begin{lstlisting}[caption={Theorem~\ref{thm:bounded-lift}: conversion to real approximate degree},
                   mathescape=true, label={lst:bounded-conversion}]
theorem bounded_integer_lift_conversion {n d : Nat}
    {f : Cube n $\to$ Bool} {P : Cube n $\to$ $\mathbb{R}$}
    (hP : P $\in$ degreeSpace (K := $\mathbb{R}$) (Fin n) d)
    (Z : Cube n $\to$ $\mathbb{Z}$) (a : $\mathbb{Z}$) (B : Nat) ($\varepsilon$ : $\mathbb{R}$)
    (h$\varepsilon$ : $\varepsilon$ $\leq$ 1 / 8)
    (hZ : $\forall$ x, a $\leq$ Z x $\land$ Z x $\leq$ a + (B : $\mathbb{Z}$))
    (herr : $\forall$ x, |P x - (Z x : $\mathbb{R}$) - bit (f x) / 2| $\leq$ $\varepsilon$) :
    realApproxDegree f (1 / 3) $\leq$ (16 * (B + 1) - 1) * d
\end{lstlisting}
\end{minipage}

Listings~\ref{lst:paturi-input} and~\ref{lst:majority-range} give the formally
verified Majority lower bound and the resulting degree--range tradeoff from
Corollary~\ref{cor:range}.

\noindent\begin{minipage}{\linewidth}
\begin{lstlisting}[caption={The proved Majority approximate-degree lower bound},
                   mathescape=true, label={lst:paturi-input}]
theorem strictMajority_realApproxDegree_lower_bound (n : Nat) :
  (1 / 10000 : $\mathbb{R}$) * n $\leq$
    (realApproxDegree (strictMajority n) (1 / 3) : $\mathbb{R}$)
\end{lstlisting}
\end{minipage}

\noindent\begin{minipage}{\linewidth}
\begin{lstlisting}[caption={Corollary~\ref{cor:range}: the Majority degree--range tradeoff},
                   mathescape=true, label={lst:majority-range}]
theorem majority_degree_range_tradeoff {n d : Nat} {P : Cube n $\to$ $\mathbb{R}$}
    (hP : P $\in$ degreeSpace (K := $\mathbb{R}$) (Fin n) d)
    (Z : Cube n $\to$ $\mathbb{Z}$) (a : $\mathbb{Z}$) (B : Nat) ($\varepsilon$ : $\mathbb{R}$)
    (h$\varepsilon$ : $\varepsilon$ $\leq$ 1 / 8)
    (hZ : $\forall$ x, a $\leq$ Z x $\land$ Z x $\leq$ a + (B : $\mathbb{Z}$))
    (herr : $\forall$ x, |P x - (Z x : $\mathbb{R}$) - bit (strictMajority n x) / 2| $\leq$ $\varepsilon$) :
    (1 / 160000 : $\mathbb{R}$) * n $\leq$ (d * (B + 1) : Nat)
\end{lstlisting}
\end{minipage}

For fixed $B$, rearranging this inequality gives
$d\ge n/(160000(B+1))$. The corresponding Lean theorem is
\lstinline|majority_fixed_width_degree_bound|.

Listing~\ref{lst:ternary-decoder} states the more general version valid for any
Boolean function $f$. Taking \lstinline|f := strictMajority n| gives
Theorem~\ref{thm:ternary}.

\noindent\begin{minipage}{\linewidth}
\begin{lstlisting}[caption={Theorem~\ref{thm:ternary}: quintic conversion for any Boolean function},
                   mathescape=true, label={lst:ternary-decoder}]
theorem ternary_lift_realApproxDegree_le {n d : Nat} (hn : 40 $\leq$ n)
    (f : Cube n $\to$ Bool) (P : Cube n $\to$ $\mathbb{R}$) (Z : Cube n $\to$ $\mathbb{Z}$)
    (hP : P $\in$ degreeSpace (K := $\mathbb{R}$) (Fin n) d)
    (hZ : $\forall$ x, -1 $\leq$ Z x $\land$ Z x $\leq$ 1)
    (happrox : $\forall$ x, |P x - (Z x : $\mathbb{R}$) - bit (f x) / 2| $\leq$ 1 / (20 * (n : $\mathbb{R}$))) :
    realApproxDegree f (1 / 3) $\leq$ 5 * d
\end{lstlisting}
\end{minipage}

The finite witness statement, Listing~\ref{lst:ternary-dual}, also holds for
arbitrary Boolean $f$ and needs only the displayed approximate-degree
inequality. Its first conclusion is precisely $\gamma\in V_d^\perp$.  The second
is the strict separation in Corollary~\ref{cor:dual}.

\noindent\begin{minipage}{\linewidth}
\begin{lstlisting}[caption={Corollary~\ref{cor:dual}: the finite strict dual witness},
                   mathescape=true, label={lst:ternary-dual}]
theorem ternary_lift_exists_strict_dual_of_degree {n d : Nat} (hn : 40 $\leq$ n)
    (f : Cube n $\to$ Bool) (Z : Cube n $\to$ $\mathbb{Z}$)
    (hZ : $\forall$ x, -1 $\leq$ Z x $\land$ Z x $\leq$ 1)
    (hd : 5 * d < realApproxDegree f (1 / 3)) :
    $\exists$ $\gamma$ : Cube n $\to$ $\mathbb{R}$,
      ($\forall$ P $\in$ degreeSpace (K := $\mathbb{R}$) (Fin n) d,
        $\sum$ x, $\gamma$ x * P x = 0) $\land$
      ($\sum$ x, |$\gamma$ x|) / (20 * (n : $\mathbb{R}$)) <
        |$\sum$ x, $\gamma$ x * ((Z x : $\mathbb{R}$) + bit (f x) / 2)|
\end{lstlisting}
\end{minipage}

Listing~\ref{lst:asymptotic-dual} matches Open Problem 2 with ties accepted and
the stated big-O degree bound. Here \lstinline|weakMajority n x| is true exactly
when $2|x|\geq n$, so it is the function $M_n$ used in the box. The notation
\lstinline|=O[atTop]| means big-O as $n\to\infty$, and $\forall^{\mathrm f} n$
means ``for all sufficiently large $n$.''  The order of quantifiers makes the
threshold independent of $Z$, while allowing the witness $\gamma$ to depend on
$Z$.

\noindent\begin{minipage}{\linewidth}
\begin{lstlisting}[caption={Corollary~\ref{cor:asymptotic-dual}: ties accepted and the big-O degree bound},
                   mathescape=true, label={lst:asymptotic-dual}]
theorem eventually_weakMajority_ternary_dual (d : Nat $\to$ Nat)
    (hd : (fun n $\mapsto$ (d n : $\mathbb{R}$)) =O[atTop]
      (fun n $\mapsto$ Real.sqrt ((n : $\mathbb{R}$) / Real.log n))) :
    $\forall^{\mathrm f}$ n : Nat in atTop, $\forall$ (Z : Cube n $\to$ $\mathbb{Z}$),
      ($\forall$ x, -1 $\leq$ Z x $\land$ Z x $\leq$ 1) $\to$
      $\exists$ $\gamma$ : Cube n $\to$ $\mathbb{R}$,
        ($\forall$ P $\in$ degreeSpace (K := $\mathbb{R}$) (Fin n) (d n),
          $\sum$ x, $\gamma$ x * P x = 0) $\land$
        ($\sum$ x, |$\gamma$ x|) / (20 * (n : $\mathbb{R}$)) <
          |$\sum$ x, $\gamma$ x * ((Z x : $\mathbb{R}$) + bit (weakMajority n x) / 2)|
\end{lstlisting}
\end{minipage}

The stronger linear degree bound is formalized separately in
Listing~\ref{lst:linear-dual}. Its hypothesis is an inequality between natural
numbers. The big-O bound above implies this inequality for all sufficiently
large $n$. Both statements are also proved with \lstinline|strictMajority| in
place of \lstinline|weakMajority|.

\noindent\begin{minipage}{\linewidth}
\begin{lstlisting}[caption={Corollary~\ref{cor:asymptotic-dual}: the eventual linear degree bound},
                   mathescape=true, label={lst:linear-dual}]
theorem eventually_weakMajority_ternary_dual_of_linear (d : Nat $\to$ Nat)
    (hd : $\forall^{\mathrm f}$ n : Nat in atTop, 50000 * d n < n) :
    $\forall^{\mathrm f}$ n : Nat in atTop, $\forall$ (Z : Cube n $\to$ $\mathbb{Z}$),
      ($\forall$ x, -1 $\leq$ Z x $\land$ Z x $\leq$ 1) $\to$
      $\exists$ $\gamma$ : Cube n $\to$ $\mathbb{R}$,
        ($\forall$ P $\in$ degreeSpace (K := $\mathbb{R}$) (Fin n) (d n),
          $\sum$ x, $\gamma$ x * P x = 0) $\land$
        ($\sum$ x, |$\gamma$ x|) / (20 * (n : $\mathbb{R}$)) <
          |$\sum$ x, $\gamma$ x * ((Z x : $\mathbb{R}$) + bit (weakMajority n x) / 2)|
\end{lstlisting}
\end{minipage}

\subsection{Explicit constructions and their guarantees}

The definition in Listing~\ref{lst:explicit-formula} evaluates
(\ref{eq:explicit-new}) on the Boolean cube. Here \lstinline|topCoefficient f|
is $c_f$, \lstinline|fullMonomial x| is $x_{[n]}$, and
\lstinline|parityCharacter x| is the value $\chi(x)$ of the alternating-sign
polynomial, with the same sign convention as in the main text. Division and
remainder in the integer expressions give $k$ and $b$, including when $c_f$ is
negative. The value \lstinline|acceptingCount f| is $|f^{-1}(1)|$, so
\lstinline|explicitError f| equals $b2^{-n-1}$.

\noindent\begin{minipage}{\linewidth}
\begin{lstlisting}[caption={The explicit approximating polynomial and its error},
                   mathescape=true, label={lst:explicit-formula}]
def explicitApproximant {n : Nat} (f : (Fin n $\to$ Bool) $\to$ Bool) (x : Fin n $\to$ Bool) : $\mathbb{R}$ :=
  (if f x then 1 else 0) / 2 -
    (topCoefficient f / 2 : $\mathbb{Z}$) * (fullMonomial x : $\mathbb{R}$) -
    (topCoefficient f % 2 : $\mathbb{Z}$) * (parityCharacter x : $\mathbb{R}$) / (2 : $\mathbb{R}$) ^ (n + 1)

def explicitError {n : Nat} (f : Cube n $\to$ Bool) : $\mathbb{R}$ :=
  if Odd (acceptingCount f) then 1 / (2 : $\mathbb{R}$) ^ (n + 1) else 0
\end{lstlisting}
\end{minipage}

Listings~\ref{lst:explicit-degree-error}
and~\ref{lst:explicit-optimal-symmetric} cover all conclusions of
Theorem~\ref{thm:explicit}. The symmetry statement preserves invariance under
each coordinate permutation \lstinline|e|.  Applying it to every permutation
proves the stated symmetry conclusion.

\noindent\begin{minipage}{\linewidth}
\begin{lstlisting}[caption={Theorem~\ref{thm:explicit}: degree and pointwise torus error},
                   mathescape=true, label={lst:explicit-degree-error}]
theorem explicitApproximant_mem_degreeSpace {n : Nat} (hn : 0 < n)
    (f : (Fin n $\to$ Bool) $\to$ Bool) :
    explicitApproximant f $\in$ degreeSpace (K := $\mathbb{R}$) (Fin n) (n - 1)

theorem explicitApproximant_torusDistance {n : Nat} (hn : 0 < n)
    (f : Cube n $\to$ Bool) (x : Cube n) :
    torusDistance (explicitApproximant f x - bit (f x) / 2) = explicitError f
\end{lstlisting}
\end{minipage}

\noindent\begin{minipage}{\linewidth}
\begin{lstlisting}[caption={Theorem~\ref{thm:explicit}: optimality and symmetry},
                   mathescape=true, label={lst:explicit-optimal-symmetric}]
theorem optimalTorusError_pred {n : Nat} (hn : 0 < n) (f : Cube n $\to$ Bool) :
    optimalTorusError f (n - 1) =
      if Odd (acceptingCount f) then 1 / (2 : $\mathbb{R}$) ^ (n + 1) else 0

theorem explicitApproximant_perm {n : Nat} (f : (Fin n $\to$ Bool) $\to$ Bool)
    (e : Equiv.Perm (Fin n)) (hf : $\forall$ x, f (x $\circ$ e) = f x) (x : Fin n $\to$ Bool) :
    explicitApproximant f (x $\circ$ e) = explicitApproximant f x
\end{lstlisting}
\end{minipage}

The formally verified coefficient bound in Listing~\ref{lst:coefficient-size}
makes the $O(n)$-bit assertion quantitative. A common denominator is $2^{n+1}$
and each integer numerator has absolute value at most $2^{2n+1}$. Here
\lstinline|mobiusCoeff| extracts a multilinear coefficient.

\noindent\begin{minipage}{\linewidth}
\begin{lstlisting}[caption={Common denominator and numerator bound},
                   mathescape=true, label={lst:coefficient-size}]
theorem explicitApproximant_coefficients_dyadic {n : Nat}
    (f : Cube n $\to$ Bool) (S : Finset (Fin n)) :
    $\exists$ N : $\mathbb{Z}$, mobiusCoeff (explicitApproximant f) S =
        (N : $\mathbb{R}$) / (2 : $\mathbb{R}$) ^ (n + 1) $\land$ |N| $\leq$ (2 : $\mathbb{Z}$) ^ (2 * n + 1)
\end{lstlisting}
\end{minipage}

The arithmetic bound is certified by Listing~\ref{lst:fast-mobius}.  The
function \lstinline|tabulate| stores a supplied truth table,
\lstinline|transform| performs the recursive Boolean M\"obius transform, and
\lstinline|read| retrieves the entry indexed by \lstinline|assignmentOfSet u|.
The cost counts integer subtractions at unit cost. It is exactly $n2^{n-1}$,
hence at most $n2^n$ in this arithmetic cost model.

\noindent\begin{minipage}{\linewidth}
\begin{lstlisting}[caption={Correctness and arithmetic cost of the fast transform},
                   mathescape=true, label={lst:fast-mobius}]
theorem transform_tabulate_spec {n : Nat} (f : Cube n $\to$ $\mathbb{Z}$) :
    ($\forall$ u, read (transform (tabulate f)) (assignmentOfSet u) = mobiusCoeff f u) $\land$
      transformCost (tabulate f) = n * 2 ^ (n - 1)
\end{lstlisting}
\end{minipage}

Listing~\ref{lst:majority-lengths} gives both parts of the corollary.  The
existential condition includes $n=1=2^0$, and the degree conclusion uses a
strict error inequality, as in the mathematical statement.

\noindent\begin{minipage}{\linewidth}
\begin{lstlisting}[caption={Corollary~\ref{cor:majority-lengths}: exact error and full-degree necessity},
                   mathescape=true, label={lst:majority-lengths}]
theorem optimalTorusError_strictMajority (n : Nat) (hn : 0 < n) :
    optimalTorusError (strictMajority n) (n - 1) =
      if $\exists$ k : Nat, n = 2 ^ k then 1 / (2 : $\mathbb{R}$) ^ (n + 1) else 0

theorem strictMajority_degree_ge_of_small_error (n : Nat) (hn : 0 < n)
    (hpow : $\exists$ k : Nat, n = 2 ^ k) {d : Nat} {$\varepsilon$ : $\mathbb{R}$}
    (h$\varepsilon$ : $\varepsilon$ < 1 / (2 : $\mathbb{R}$) ^ (n + 1))
    (h : HasTorusApproximation (strictMajority n) d $\varepsilon$) : n $\leq$ d
\end{lstlisting}
\end{minipage}

\end{document}